\documentclass[format=acmsmall, review=false]{acmart}
\pdfoutput=1
\usepackage{acm-ec-18}
\usepackage{natbib,booktabs}
\usepackage{wrapfig,subcaption}
\usepackage{enumerate,setspace}
\usepackage{amsmath,amssymb}

\usepackage[boxed]{algorithm}
\usepackage[noend]{algorithmic}
\usepackage{graphicx}
\usepackage{tikz}

\usetikzlibrary{decorations.pathreplacing, decorations.pathmorphing, positioning, calc}
\usetikzlibrary{arrows,shapes,  fit, positioning}
\usepackage{tkz-graph}

\newcommand{\bbR}{\mathbb{R}_{\geq 0}}
\newcommand{\bbN}{\mathbb{N}}

\newcommand{\calA}{\mathcal{A}}

\newcommand{\calN}{\mathcal{N}}
\newcommand{\calS}{\mathcal{S}}

\numberwithin{equation}{section}
\allowdisplaybreaks

\usepackage{hyperref}

\begin{document}
	\title{Robustness against Agent Failure in Hedonic Games}
	\author{Ayumi Igarashi}
	\affiliation{Kyushu University, Fukuoka, Japan}
	\author{Kazunori Ota}
	\affiliation{Kyushu University, Fukuoka, Japan}
	\author{Yuko Sakurai}
	\affiliation{National Institute of Advanced Industrial Science and Technology, Tokyo, Japan}
	\author{Makoto Yokoo}
	\affiliation{Kyushu University, Fukuoka, Japan}
\begin{abstract}
We study how stability can be maintained even after any set of at most $k$ players
leave their groups, in the context of hedonic games. While stability properties
ensure an outcome to be robust against players' deviations, it has not been considered how an
unexpected change caused by a sudden deletion of players affects
stable outcomes. In this paper, we propose a novel criterion that reshapes stability
form robustness aspect. We observe that some stability properties can be no longer preserved even when a single agent is removed.
However, we obtain positive results by focusing on symmetric
friend-oriented hedonic games. We prove that we can efficiently decide
the existence of robust outcomes with respect to Nash stability under
deletion of any number of players or contractual individual
stability under deletion of a single player. We also show that symmetric additively separable games always admit an individual stable outcome
that is robust with respect to individual rationality.
	\vspace{-10pt}
\end{abstract}

\maketitle

\section{Introduction}	\label{sec:intro}
Coalition formation is everywhere in human activities. Companies group their employers into project teams. Countries form coalitions to promote international trade among them. Individuals interact with each other and form groups in order to achieve objectives they cannot seek for on their own. 

{\em Hedonic coalition formation games} (for short, hedonic games), introduced by \citet{Bogomolnaia2002} and \citet{Banerjee2001}, provide an elegant framework to formulate coalition formation. In these games, each player has preferences over the coalitions to which she or he belongs, and desirable outcomes often correspond to {\em stable} partitions. The basic intuition behind stable partitioning is that group structures need to be robust under certain changes {\em within} the system; that is, outcomes must be immune to players' coalitional or individual deviations to other coalitions. Stability can prevent internal conflicts among members of a coalition, or it can help us nurture the relationship among members of a project team.

In many real-world scenarios, however, groups may encounter unexpected changes and challenges, imposed from the {\em outside} of the system. For instance, a certain country can go bankrupt and be enforced to leave a political alliance. In this respect, a group structure that satisfies a standard stability requirement can become immediately unstable due to unexpected circumstances. A case in point is a political coalition of three countries with one intermediate country connecting two other countries who are enemies to each other: if the intermediate player happens to disappear from the coalition, one cannot maintain the stability of the whole system.

In this paper, we propose a novel criterion that redefines stability from robustness aspect. We define an outcome to be {\em robust} with respect to a certain stability requirement $\alpha$ if removing any set of at most $k$ players still preserves $\alpha$. Besides the preceding example of a political alliance, there are several applications of hedonic games, such as project team formation \citep{okimoto2015}, research team formation \citep{Alcalde2004}, and group activity selection \citep{Darmann2012}, in which unexpected players' non-participation may severely affect stability of the system. To the best of our knowledge, however, no attempt has been ever made to connect two important considerations, robustness and stability. Our goal is to make the first step filling this gap. 
\smallskip

\noindent
{\bf Our contribution}
We focus on friend-oriented games, introduced by \citet{Dimitrov2006}, where players' preferences are succinctly encoded via the binary friendship relations. While it is known that such games always guarantee the existence of stable outcomes, we observe that a simple example of one player connecting two enemies shows impossibility in maintaining most of the stability properties, such as {\em core stability}, {\em Nash stability}, {\em individual stability}, and {\em contractual individual stability}. Not surprisingly, this negative result holds even under a very small change of the system, i.e., only a single player can disappear. 

Given these non-existence results, we investigate the computational complexity of deciding the existence of a robust outcome in a symmetric friend-oriented game. Specifically, we show that we can efficiently decide the existence of an outcome that is robust with respect to Nash stability, irrespective of the number of players leaving the game. We then prove that any symmetric friend-oriented game admits a polynomial time algorithm that finds a robust outcome with respect to contractual individual stability in case of removing a single player. To this end, we obtain a non-trivial characterization of games whose corresponding robust outcomes are non-empty. Moreover, we complement this result by showing that the problem becomes NP-hard when $k =2$. We also show that the positive results do not extend to an intermediate stability property, individual stability: we prove that the associated problem for individual stability is NP-hard even when only a single player is allowed to leave. 

Finally, we consider the question of whether a minimum stability requirement, {\em individual rationality}, can be maintained while ensuring that an outcome of a game itself satisfies stronger stability desiderata. It turns out that when players have symmetric additively separable preferences, an individually stable partition which is robust with respect to individual rationality always exist. 
Our complexity results are summarized in Table \ref{table}.
\begin{table}[t]
	\centering
	\begin{tabular}{ll}
	\toprule
		NS-robustness & poly time (Cor.~\ref{cor:NS}) \\
		IS-robustness & NP-complete ($k=1$) (Th.~\ref{thm:NPh:IS})\\
		CIS-robustness & poly time ($k =1$) (Th.~\ref{thm:CIS:poly}) \\
                                  & NP-complete ($k = 2$) (Th.~\ref{thm:NPh:CIS}) \\
		IS \& IR-robustness & exists and polytime (Th.~\ref{thm:sF:IS-IRrobust}) \\
	\bottomrule
	\end{tabular}
	\caption{Overview of our complexity results in a symmetric friend-oriented game, where $k$ is the maximum number of players who can leave the entire game.}
	\vspace{-5pt}
	\label{table}
\end{table}
\smallskip

\noindent
{\bf Related work} Several papers considered robustness against agent failures in the context of cooperative games. \citet{bachrach2011} proposed the reliability extension of cooperative games where each agent has an independent failure probability. This probabilistic model has been also applied to subclasses of cooperative games, such as totally-balanced games~\citep{bachrach2012}, weighted voting games~\citep{bachrach2013}, and cooperative max-games~\citep{bachrach2014}. 
\citet{okimoto2015} introduced the concept of {\em $k$-robustness} for team formation problems; under their definition, each team still needs to accomplish their task even after $k$ agents fail. 

Crudely, two different approaches deal with cooperative games with uncertainty. One is the Bayesian approach assuming
a known prior over agent capabilities \citep{chalkiadakis2004coalition, chalkiadakis2007core,
myerson2007core}. Another approach, initiated by \citet{balcan2015learning}, is to apply the PAC (probably approximately correct) learning model to cooperative games.
Specifically, \citet{balcan2015learning} studied the learnability of probable stable payoffs, given random samples of coalitions, which has been extended to hedonic games 

Our work is also related to the rich body of the literature on the study of hedonic games. \citet{Bogomolnaia2002} and \citet{Banerjee2001} were the first to model hedonic coalition formation games in which players' preferences solely depend on the members of each coalition. \citet{Bogomolnaia2002} considered various possibilities of players' deviations, which gives rise to different concepts of stability outcomes. Several important subclasses of hedonic games have been later proposed, including additively separable hedonic games (ASHGs) \citep{Bogomolnaia2002}, friend and enemy oriented games \citep{Dimitrov2006}, fractional hedonic games \citep{aziz2014fractional,AzizBBHOP17}, to name a few.

Several works explored the relation between stability and the networks capturing agents' preferences, in which the nodes of a graph represent players and edges correspond to the degree of preference. 
\citet{bilo1} analyzed the ratio between the social welfare of a Nash stable outcome and social optimum in fractional hedonic games, for different topologies such as bipartite graphs and trees.
In a more general setting, \citet{peters2016graphical} introduced graphical hedonic games and obtained a number of computational complexity results of stability outcomes for the case when the graph describing agents' preferences has a bounded treewidth. \citet{igarashi2016hedonicgraph} used a different approach and considered hedonic games where players are located on a graph and coalitions are only allowed to form if they are connected in this graph; they proved both existence and complexity results of some stability concepts on acyclic graphs.

Our definition of robustness is arguably the most stringent requirement one could aim for, as it requires an outcome to be immune to {\em any} possibility of deterministic agent failure. However, the definition resembles some graph connectivity concepts, such as the $k$-vertex-connectivity, capturing the robustness of a given network (see, e.g., \citet{Schrijver2003}). 
We also note that the notion of robustness in a stable matching is fundamentally different from ours. For instance, \citet{kojima2011} considers instability that results from agents' manipulation; a mechanism is considered to be robustly stable if it is strategy-proof and immune to a blocking pair before and after an agent misrepresents her preference. This does not take into account the possibility of making the system unstable due to agent failures. 

\section{Preliminaries}
For a natural number $s \in \bbN$, we write $[s]=\{1,2,\ldots,s\}$. A hedonic game is defined as a pair $(N,(\succeq_{i})_{i \in N})$ where $N=[n]$ is a finite set of {\em players} and each $\succeq_{i}$ is a preference over the subsets of $N$ (also referred to as {\em coalitions}); specifically for every $i \in N$, we let $\calN_i$ denote the collection of all coalitions containing $i$; each $\succeq_{i}$ describes a complete and transitive preference over the sets in $\calN_i$. Let $\succ_{i}$ denote the strict preference derived from $\succeq_{i}$, i.e., $S \succ_i T$ if $S \succeq_i T$, but $T \not \succeq_{i} S$. For $i \in N$ and $S,T \in \calN_i$, we say that player $i$ {\em strictly prefers} a coalition $S$ to another coalition $T$ if $S \succ_i T$; player $i$ {\em weakly prefers} $S$ to $T$ if $S \succeq_i T$. 
We call a coalition $S \subseteq N$ {\it individually rational} if every player $i \in S$ weakly prefers $S$ to $\{i\}$. 

A preference profile $(\succeq_{i})_{i \in N}$ is said to be {\it additively separable} if there exists a {\em weight function} $w:N \times N 
\rightarrow \bbR$ such that for each $i \in N$ and each $S,T\in\calN_i$ 
we have $S \succeq_{i} T$ if and only if $\sum_{j \in S}w(i,j) \ge \sum_{j \in T}w(i,j)$~\citep{Bogomolnaia2002}; we will assume that
$w(i,i)=0$ for each $i \in N$. An additively separable preference is said to be {\it symmetric} if the weight function $w:N \times N \rightarrow \bbR$ is symmetric, i.e., $w(i,j)=w(j,i)$ for all $i,j \in N$. We use the notation $(N,w)$ to denote an additively separable game with weight function $w:N \times N \rightarrow \bbR$. 
For additively separable games, each player can consider every other player to be either a friend, a neutral player, or an enemy; specifically, for each pair of distinct players $i,j \in N$, we say that $j$ is a {\em friend} of $i$ if $w(i,j) > 0$, and $j$ is an {\em enemy} of $i$ if $w(i,j) < 0$.

\citet{Dimitrov2006} introduced a subclass of additively separable preferences, which they called {\it friend-oriented preferences}. Under friend-oriented preferences, each player has strong favour towards her friends: $w(i,j)\in \{n,-1\}$ for each $i,j \in N$ with $i \neq j$. For a symmetric additively separable game $(N,w)$, let $G_w$ denote the {\em friendship graph} where the set of vertices is given by the set of players and two players $i,j$ are adjacent if and only if they are friends; each coalition $S$ is said to have {\em minimum degree} $t$ if each player in $S$ has at least $t$ other friends in $S$. 

An {\em outcome} of a hedonic game is a partition of players into disjoint coalitions. Given a partition $\pi$ of $N$ and 
a player $i \in N$, let $\pi(i)$ denote the unique coalition in $\pi$ that contains $i$. Much of the existing literature is concerned with outcomes that satisfy certain stability requirements. 
A minimum stability property we require is {\em individual rationality}. 
A partition $\pi$ of $N$ is said to be {\it individually rational} (IR) if each player prefers their coalition to staying alone, i.e., all coalitions in $\pi$ are individually rational.
If we extend this to a group deviation, we obtain the definition of the {\em core}.
Specifically, a coalition $S \subseteq N$ {\it strongly blocks} a partition $\pi$ of $N$ if every player $i \in S$ strictly prefers $S$ to her own coalition $\pi(i)$.
A partition $\pi$ of $N$ is said to be {\it core stable} (CR) if no coalition $S \subseteq N$ 
strongly blocks $\pi$.
We also consider deviations based on individual movements. Specifically, consider a player $i \in N$ and a pair of coalitions 
$S\not\in\calN_i$, $T\in\calN_i$. 
A player $j \in S$ {\it accepts} a deviation of $i$ to $S$ 
if $j$ weakly prefers $S\cup \{i\}$ to $S$; a player $j \in T$ {\it accepts} a deviation of $i$ to $S$ 
if $j$ weakly prefers $T\setminus \{i\}$ to $T$.  
A deviation of $i$ from $T$ to $S$ is an {\it NS-deviation} if $i$ strictly prefers $S \cup \{i\}$ to $T$, an {\it IS-deviation} if it is an NS-deviation and all players in $S$ accept it, and a {\em CIS-deviation} if it is an IS-deviation and all players in $T$ accept it.
A partition $\pi$ is called {\em Nash stable} (NS) (respectively, {\em individually stable} (IS) and {\em contractually individually stable} (CIS)) if no player $i\in N$ has an NS-deviation
(respectively, an IS-deviation and a CIS-deviation) from $\pi(i)$ to another coalition
$S\in \pi$ or to $\emptyset$. 

Trivially, Nash stability implies individual stability, which also implies contractually individual stability. Usually, core stability does not imply the stability based on individual deviations. However, we note that for a symmetric friend-oriented game, core stability implies individual stability. 

\begin{lemma}\label{lem:relation}
For a symmetric friend-oriented game $(N,w)$, if a partition $\pi$ is core stable, then $\pi$ is individually stable. 
\end{lemma}
\begin{proof}
Suppose towards a contradiction that $\pi$ is core stable but does not satisfy individual stability. Then, there is a player $i$ who has an IS-deviation to some coalition $S \in \pi$. This means that $i$ strictly prefers $S \cup \{i\}$ to $\pi(i)$, and $i$ and every player in $S$ are friends to each other. Thus, the players $S \cup \{i\}$ strongly block $\pi$, a contradiction. 
\end{proof}

We also note that contractual individual stability does not normally imply individual rationality as a player's deviation under the stability concept needs to be approved by the members of her coalition. With symmetric friend-oriented preferences, nevertheless, the implication holds. 

\begin{lemma}\label{lem:relation}
For a symmetric friend-oriented game $(N,w)$, if a partition $\pi$ is contractually individually stable, then $\pi$ is individually rational. 
\end{lemma}
\begin{proof}
Suppose towards a contradiction that $\pi$ is contractually individually stable but does not satisfy individual rationality. Then, there is a player $i$ who strictly prefers $\{i\}$ to his own coalition $\pi(i)$. This means that $\pi(i)$ contains at least two players and every other player $j \in \pi(i) \setminus \{i\}$ is an enemy of $i$. Thus, $i$ has a CIS-deviation to the emptyset, a contradiction. 
\end{proof}

\section{Agent failure in hedonic games}
Earlier we defined the robustness informally: A sudden deletion of players upon an outcome should preserve the property it has achieved before. We are now in a position to make the definition more formal. For each $S \subseteq N$ and $i \in N$, we denote by $\succeq_i|_{S}$ the preference relation restricted to $\calN_i \cap 2^S$. 

\begin{definition}
Given $\alpha \in \{\mbox{CR,NS,IS,CIS, IR}\}$ and a natural number $k>0$, a partition $\pi$ is said to be $\alpha$-robust under deletion of at most $k$ players if $\pi$ satisfies the property $\alpha$, and for any $S \subseteq N$ with $|S| \le k$, the partition $\pi_{-S} :=\{\, S'\setminus S \mid S' \in \pi\,\}$ still satisfies the property $\alpha$ in the subgame $(N\setminus S,(\succeq_{i}|_{N\setminus S})_{i \in N \setminus S})$. When $k$ is clear from the context, we will simply call such partition $\alpha$-robust.
\end{definition}

By definition, if an outcome is $\alpha$-robust under deletion of $k+1$ players, then it is $\alpha$-robust under deletion of any $\ell \leq k$ players. 
Also, fixing parameter $k$, the relations between the above robustness concepts are the same as those among the corresponding stability concepts. Namely, we have the following containment relation among the classes of outcomes: $\mbox{NS-robust} \subseteq \mbox{IS-robust}  \subseteq \mbox{CIS-robust}  \cap \mbox{IR-robust}$,  and $\mbox{CR-robust} \subseteq \mbox{IR-robust}$. Also, by Lemma \ref{lem:relation}, CR-robustness implies IS-robustness and CIS-robustness implies IR-robustness for a symmetric friend-oriented game. 

It is known that a stable outcome of a symmetric friend-oriented game is guaranteed to exist and can be found in polynomial time: a partition that divides the players into the connected components satisfies the preceding stability requirements. However, the example below illustrates that even when players have symmetric friend-oriented preferences, an $\alpha$-robust partition under deletion of a single player may not exist for any $\alpha \in \{CR,NS,IS, CIS\}$. 

\begin{example}\label{ex:friend:nonexistence}
\upshape
Consider a symmetric friend-oriented game $(N,w)$ with three players $a$, $b$, and $c$. The friendship graph $G_w$ forms a star with the center being $b$ (Figure \ref{fig:friend}). 
\begin{figure}[htb]
\centering
\begin{tikzpicture}[scale=1, transform shape]
	\node[draw, circle](1) at (-1.5,0) {$a$};
	\node[draw, circle](2) at (0,0) {$b$};
	\node[draw, circle](3) at (1.5,0) {$c$};
	
	\draw[-, >=latex,thick] (1)--(2) (2)--(3);	
\end{tikzpicture}
\caption{Non-existence of a CR-robust partition for symmetric friend-oriented games.
\label{fig:friend}
}
\end{figure}
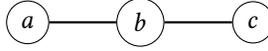

Suppose towards a contradiction that $\pi$ is CIS-robust under deletion of a single player. 
\begin{itemize}
\item First, suppose $\pi = \{\{a,b,c\}\}$. Then, without $b$, the coalition is not individually rational, a contradiction. 
\item Second, suppose $\pi = \{\{a\},\{b,c\}\}$ or $\pi = \{\{a,b\},\{c\}\}$. Then, if the player who belongs to the same coalition as $b$ disappears, player $b$ would have a CIS-deviation to the other coalition, a contradiction. 
\item Third, if $\pi = \{\{a\},\{b\},\{c\}\}$, it would not satisfy contractually individual stability, a contradiction. 
\item Finally, if $\pi=\{\{a,c\},\{b\}\}$, then $\pi$ would not be individually rational, a contradiction. 
\end{itemize}
We have exhausted all possible cases and obtained a contradiction. Hence the game admits no CIS-robust partition. This means that the game does not have an $\alpha$-robust outcome for any $\alpha \in \{CR,NS,IS,CIS\}$.\qed
\end{example}

\section{NS-robustness}
We saw that a symmetric friend-oriented game may not admit an NS-robust outcome. In this section, we show that deciding the existence of an NS-robust outcome remains easy for a symmetric friend-oriented game. We warm up by observing that in order to preserve individual rationality, each coalition must be a clique or have minimum degree at least $k+1$. 
\begin{lemma}\label{lem:IR}
For any symmetric friend-oriented game, $k>0$, and any IR-robust partition $\pi$, each $S \in \pi$ is either a clique or has minimum degree at least $k+1$. 
\end{lemma}
\begin{proof}
Let $\pi$ be an IR-robust partition. Suppose towards a contradiction that there is a coalition $S \in \pi$ such that $S$ does not form a clique and there is a player $i \in S$ who has at most $k$ friends in $S$. This means that by IR-robustness, $S$ has size at most $k+1$; otherwise, removing all $i$'s friends in $S$ would violate individual rationality for $i$. Now since $S$ is not a clique, there is a player $j$ who has an enemy in $S$. Observe that $j$ has at most $k-1$ friends in $S$, since $S$ has size at most $k+1$ and at least one of the players is an enemy of $j$. Hence if all the friends of $j$ in $S$ disappear, this would cause the deviation of $j$ to staying alone, contradicting IR-robustness.
\end{proof}

Observe that if there is a coalition of size at most $k+1$ and some player has a friend in other coalitions, the player would have an NS-deviation to the other coalition after removal of $k$ players. Hence, any NS-robust outcome cannot contain such coalition, which leads to the following characterization of the classes of friend-oriented games whose NS-robust outcomes are non-empty. 

\begin{theorem}\label{thm:NS:friend}
The following conditions are equivalent for any symmetric friend-oriented game $(N,w)$ and any natural number $k>0$:
\begin{enumerate}
\item There exists an NS-robust partition. 
\item Each connected component of $G_w$ is either a clique or has minimum degree at least $k+1$. 
\end{enumerate}
\end{theorem}
\begin{proof}
Suppose towards a contradiction that there exists an NS-robust outcome $\pi$ but for some connected component $G'=(V',E')$ of $G_w$, $G'$ is not a clique and there is a player $i$ with at most $k$ friends. Let $S=\pi(i)$. By Lemma \ref{lem:IR}, $S$ is a clique of size at most $k+1$. 
Now since $G'$ is not a clique, $V' \setminus S$ is non-empty. By connectivity of $G'$, there is a player $j \in S$ having a neighbor $i' \in V' \setminus S$. But this implies that when all the players in $S$ except for $j$ disappear, player $j$ would have an incentive to deviate to the coalition of $i'$, a contradiction.

Conversely suppose that each connected component of $G_w$ is a clique, or has minimum degree at least $k+1$. Let $\pi$ be a partition that divides the players into the connected components of $G_w$. Clearly, no player has an incentive to deviate to another coalition at $\pi$. Also, removing at most $k$ players does not affect Nash stability. Indeed, after removal of at most $k$ players, each player $i \in N$ has at least one friend in his coalition of the resulting partition and has no friend in the other coalitions; or $i$ forms a singleton and has no friend in the other coalitions.
Hence, $\pi$ is NS-robust.
\end{proof}

\begin{corollary}\label{cor:NS}
For a symmetric friend-oriented game $(N,w)$, deciding the existence of NS-robust outcomes can be done in polynomial time. 
\end{corollary}
\begin{proof}
The condition $(2)$ of Theorem \ref{thm:NS:friend} can be easily verified by checking the size and minimum degree of each connected component, which can be done in $O(n+m)$ by depth-first search, where $m$ is the number of edges in $G_w$.  
\end{proof}

\section{CIS-robustness and IS-robustness}
We now turn our attention to a weaker stability concept, {\em contractually individual stability}. Usually, such stability requirement is not difficult to achieve: \citet{Gairing2010} observed that a CIS partition is guaranteed to exist for any symmetric additively separable game and can be efficiently computed. As we have seen before, the presence of a star with two leaves complicates the existence of CIS-robust outcomes. In what follows, we will show that by decomposing the friendship graph appropriately, one can determine the existence of CIS-robust outcomes under deletion of a single player. We start by showing that a leaf player and its unique neighbor playing a role of {\em pseudo-center} form a pair in a CIS-robust outcome. For a graph $G=(V,E)$ and a subset $X \subseteq V$, we denote by $G \setminus X$ the subgraph of $G$ induced by $V \setminus X$. We say that a vertex $j$ is a {\em pseudo-center} in a graph $G$ if at most one neighbor of $j$ is a non-leaf vertex. 

\begin{lemma}\label{lem2:IR}
For a symmetric friend-oriented game $(N,w)$ and $k=1$, let $\pi$ be an arbitrary CIS-robust partition. If $j$ is the unique friend of $i$, and $j$ is a pseudo-center in $G_w$, then $\pi(i)=\{i,j\}$. 
\end{lemma}
\begin{proof}
By Lemma \ref{lem:IR}, $i$'s coalition is either the singleton $\{i\}$ or the pair $\{i,j\}$. 
Assume towards a contradiction that $\pi(i)=\{i\}$. 
If $|\pi(j)| =1$, then $i$ has a CIS-deviation to $j$'s coalition, a contradiction. 
If $|\pi(j)| =2$, then removing the player $h \neq j$ in $\pi(j)$ would cause the CIS-deviation of $i$ to $j$'s coalition, a contradiction. 
If $|\pi(j)| \geq 3$, then this means that $j$ has at least two friends $a,b$ in $\pi(j)$ by Lemma \ref{lem:IR}. However, this means that at least one of the players $a,b$ has only one friend $j$ but has at least one enemy in $\pi(j)$, contradicting Lemma \ref{lem:IR}. In either case, we obtain a contradiction. 
\end{proof}

The above lemma can recursively apply to all such pairs in the following way: as long as there is an edge $\{i,j\}$ satisfying the property in Lemma \ref{lem2:IR}, we need to put the players into a pair and examine whether such an edge still exists in the remaining instance. This allows us to partially determine the structure of a CIS-robust outcome. Figure \ref{fig:fixpair} illustrates the sequence of pairs of players that need to be formed in a robust outcome. We now formalize the above idea as follows. For a friendship graph $G_w$, a sequence of edges $\{i_t,j_t\}$ for $t=1,2,\ldots,t^*$ is called an {\em outer elimination sequence} if the following two hold:  
\begin{itemize}
\item[$(${\rm E1}$)$] $j_t$ is the unique friend of $i_t$ in $G_t$, and $j_t$ is a pseudo-center in $G_t$; or
\item[$(${\rm E2}$)$] $j_t$ is the unique friend of $i_t$ in $G_t$, and $i_t$ is a friend of some player in $\bigcup^{t-1}_{h=1}\{i_h,j_h\}$. 
\end{itemize}
Here $G_t=G_w \setminus \bigcup^{t-1}_{h=1}\{i_h,j_h\}$ for each $t=1,2,\ldots,t^*$. An outer elimination sequence $(\{i_t,j_t\})_{t=1,2,\ldots,t^*}$ is said to be {\em maximal} if it cannot be made any longer, i.e., there is no outer elimination sequence $(\{i_t,j_t\})_{t=1,2,\ldots,t^*+1}$.

\begin{lemma}\label{cor:pair}
For a symmetric friend-oriented game $(N,w)$ and $k=1$, let $\pi$ be an arbitrary CIS-robust partition. If there is an outer elimination sequence $\{i_t,j_t\}$ for $t=1,2,\ldots,t^*$, then we have $\pi(i_t)=\{i_t,j_t\}$ for each $t=1,2,\ldots,t^*$. 
\end{lemma}
\begin{proof}
We prove the statement by induction on $t$. When $t=1$, the claim holds due to Lemma \ref{lem2:IR}. Suppose that the claim holds for $t \leq d-1$ and we prove it for $t=d$. Now by the induction hypothesis, $\pi(i_h)=\{i_h,j_h\}$ for each $h=1,2,\ldots,t-1$, and thus players $i_{t}$ and $j_{t}$ form a coalition within $G_{t}$. Now, we have either $\pi(i_{t})=\{i_{t}\}$ or $\pi(i_{t})=\{i_{t},j_{t}\}$ by Lemma \ref{lem:IR}. Assume towards a contradiction that $\pi(i_{t})=\{i_{t}\}$.
First suppose that $j_{t}$ is a pseudo-center in $G_t$. Again, if $|\pi(j_t)| \geq 3$, then $\pi(j_t)$ contains at least two friends $a,b$ of $j_t$ by Lemma \ref{lem:IR} where at least one of the players has only one friend $j_t$ in $\pi(j_t)$, contradicting Lemma \ref{lem:IR}. Thus, $j_{t}$ either stays alone at $\pi$ or forms a coalition with his another friend in $G_{t}$; however, in the former case, $i_{t}$ would have a CIS-deviation to $j_{t}$; and in the latter case, deleting the other friend of $j_{t}$ would cause the CIS-deviation of $i_{t}$ to $\pi(i_{t})$, a contradiction.  
Second, suppose that $i_{t}$ is a friend of some player $i \in \bigcup^{t-1}_{h=1}\{i_h,j_h\}$. Then, by removing $j \in \pi(i)$ with $j \neq i$, player $i$ would have a CIS-deviation to $\pi(i_{t})$, a contradiction. We thus conclude that $\pi(i_{t})=\{i_{t},j_{t}\}$. 
\end{proof}

\begin{figure}[htb]
\centering
\begin{tikzpicture}[scale=0.8, transform shape, every node/.style={minimum size=6mm, inner sep=1pt}]
	\node[draw, circle](1) at (-3,0) {$i_1$};
	\node[draw, circle](2) at (-2,0) {$j_1$};
	\node[draw, circle](3) at (-1,0) {$i_2$};
	\node[draw, circle](4) at (0,0) {$j_2$};
	
	\node[draw, circle](5) at (1,1) {$i_3$};
	\node[draw, circle](7) at (1,-1) {$i_4$};
	\node[draw, circle](6) at (2.2,1) {$j_3$};
	\node[draw, circle](8) at (2.2,-1) {$j_4$};
	\draw[-, >=latex,thick] (1)--(2) (2)--(3) (3)--(4) (4)--(5) (6)--(5) (8)--(7) (4)--(7) (8)--(6);

	\draw (-2.5,0) ellipse (28pt and 20pt);
	\draw (-0.5,0) ellipse (28pt and 20pt);
	\draw (1.6,1) ellipse (28pt and 20pt);
	\draw (1.6,-1) ellipse (28pt and 20pt);
	
	\begin{scope}[xshift=8cm]
	\node[draw, circle](1) at (-3,0) {$i_1$};
	\node[draw, circle](2) at (-2,0) {$j_1$};
	\node[draw, circle](3) at (-1,0) {$i_2$};
	\node[draw, circle](4) at (0,0) {$j_2$};
	
	\node[draw, circle](5) at (1,1) {$i_3$};
	\node[draw, circle](7) at (1,-1) {$s$};
	\node[draw, circle](6) at (1,0) {$j_3$};
	\draw[-, >=latex,thick] (1)--(2) (2)--(3) (3)--(4) (4)--(5) (6)--(5) (6)--(7) (7)--(4);

	\draw (-2.5,0) ellipse (26pt and 18pt);
	\draw (-0.5,0) ellipse (26pt and 18pt);
	\draw (1,0.5) ellipse (18pt and 28pt);
	\end{scope}
\end{tikzpicture}
\caption{Examples of maximal outer elimination sequences. Observe that the left friendship graph admits a unique CIS-robust partition that consists of elimination pairs, whereas the right friendship graph admits no CIS-robust partition since deleting $i_3$ would cause the CIS-deviation of $j_3$ to $s$.}
\label{fig:fixpair}
\end{figure}
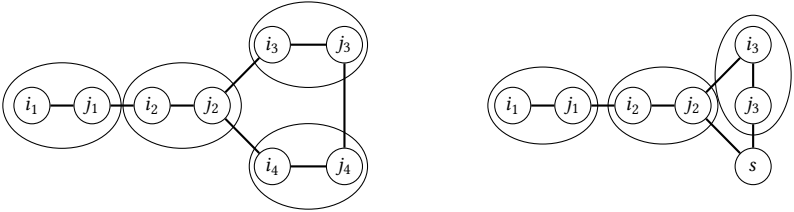

Given a symmetric friend oriented game $(N,w)$, we say that a pair of players is an {\em elimination pair} if it appears in some outer elimination sequence. We denote by $P_w$ the set of players who belong to some elimination pair, by $S_w$ the set of players who have exactly one friend in $N \setminus P_w$, by $B_w$ the set of players who have at least two friends in $N \setminus (P_w \cup S_w)$, and by $R_w$ the set of remaining players, i.e., $R_w=N \setminus (P_w \cup S_w \cup B_w)$. 
Before we proceed, we observe the following. 

\begin{lemma}\label{lem}
For a symmetric friend oriented game $(N,w)$, each player in $S_w$ has no friend in $P_w$. 
\end{lemma}
\begin{proof}
Suppose that there is a player $i \in S_w$ who is a friend of some player in $P_w$. Let $j$ be the unique friend of $i$ in $N \setminus P_w$. Then, the pair of players $i$ and $j$ is an elimination pair satisfying the condition $(${\rm E2}$)$, a contradiction. 
\end{proof}

\begin{lemma}\label{lem:degree}
For a symmetric friend oriented game $(N,w)$, let $j \in R_w$. Then, $j$ has no friend in $N \setminus P_w$. Further if there is a CIS-robust outcome, $j$ has no friend in $P_w$. 
\end{lemma}
\begin{proof}
Consider $j \in R_w$. Assume towards a contradiction that $j$ has some friend in $N \setminus P_w$.
If $j$ has a friend $i \in S_w$, then $i$ together with $j$ is an elimination pair and must be included in $P_w$, a contradiction. 
If $j$ has exactly one friend in $N \setminus (P_w \cup S_w)$, then this means that $j \in S_w$, a contradiction.
If $j$ has at least two friends in  $N \setminus (P_w \cup S_w)$, then this means that $j \in B_w$, a contradiction. Hence $j$ has no friend in $N \setminus P_w$. 

Further assume otherwise that there is a CIS-robust outcome $\pi$ but there is a player $i \in P_w$ who is a friend of $j$. 
By Lemma \ref{cor:pair}, $i$ forms a pair, say with player $h$ at $\pi$. 
By Lemma \ref{cor:pair} and Lemma \ref{lem:IR}, $j$ stays alone at $\pi$. Thus deleting $h$ would cause the CIS-deviation of $i$ to $j$, a contradiction. 
\end{proof}

Figure \ref{fig:decomposition} illustrates the partition of the player set into $P_w,S_w,B_w$, and $R_w$. Now, a CIS-robust outcome must include all the elimination pairs, and hence if such outcome exists, there is at most one maximal outer elimination sequence. We can thus completely characterise the class of symmetric friend-oriented games that admit a CIS-robust outcome under deletion of a single player. 

\begin{theorem}\label{thm:CIS}
For a symmetric friend-oriented game and $k=1$, a CIS-robust outcome exists if and only if the following holds:
\begin{itemize}
\item[$(${\rm i}$)$] the set of elimination pairs that appear in each maximal elimination sequence is the same; and
\item[$(${\rm ii}$)$] there are no elimination pairs $\{i,j\}$ and $\{u,v\}$ where $i$ is a friend of both $u$ and $v$; and
\item[$(${\rm iii}$)$] for each player $i \in P_w$ and each player $j \in R_w$, $i$ and $j$ are enemies to each other; and
\item[$(${\rm iv}$)$] if there is a player $i \in P_w$ who is a friend of every player in $B_w$, then every player $j \in B_w$ is a friend of exactly one player in $S_w$ and an enemy of at least one player in each elimination pair.  
\end{itemize}
\end{theorem}
\begin{proof}
Suppose that there is a CIS-robust outcome $\pi$. To show $(${\rm i}$)$, take any maximal elimination sequence $(\{i_t,j_t\})_{t=1,2,\ldots,t^*}$ and $(\{a_h,b_h\})_{h=1,2,\ldots,s^*}$. If there are two elimination pairs $\{i_t,j_t\}$ and $\{a_h,b_h\}$ with $i_t=a_h$ and $j_t \neq b_h$, this would imply that $\pi(i_t)=\{i_t,j_t\}=\{i_t,b_h\}$ by Lemma \ref{cor:pair}, a contradiction. If the two sequences are disjoint, then one can create a longer elimination sequence by adding edge $\{a_1,b_1\}$ to the last position of the other sequence, contradicting maximality. 
To see  $(${\rm ii}$)$, if there are two pairs $\{i,j\}$ and $\{u,v\}$ where $i$ is a friend of both $u$ and $v$, then $\pi$ has to include both pairs, which however implies that $i$ would have a CIS-deviation to the coalition $\{u,v\}$ after the removal of player $j$, a contradiction. 
The statement $(${\rm iii}$)$ holds due to Lemma \ref{lem:degree}.
To see $(${\rm iv}$)$, assume that some player $i \in P_w$ is a friend of every player in $B_w$. Consider any player $j \in B_w$. If $\pi(j) \subseteq B_w$, then by deleting the other friend $h \neq i$ with $h \in \pi(i)$, $i$ would have a CIS-deviation to $\pi(j)$, a contradiction. Thus we have $\pi(j) \not \subseteq B_w$. Further, each player in $R_w$ has no friend and hence stays alone at $\pi$ by Lemma \ref{lem:IR}. This means $\pi(j)\cap R_w=\emptyset$, and thus $\pi(j) \cap S_w \neq \emptyset$. However, if $j$ has no friend in $S_w$ or multiple friends in $S_w$, player in $S_w$ who belongs to $j$'s coalition has no friend in $\pi(j)$ or has at most one friend and at least one enemy in $\pi(j)$, contradicting Lemma \ref{lem:IR}. Hence $j$ has exactly one friend $s$ in $S_w$; by Lemma \ref{lem:IR} and by the fact that $\pi(j) \cap S_w \neq \emptyset$, we have $\pi(j)=\{j,s\}$. If there is an elimination pair $\{u,v\}$ where both of them are adjacent to $j$, then $j$ would have a CIS-deviation to the coalition $\{u,v\}$ after removal of $s$. Hence, $j$ is an enemy of at least one player in each elimination pair. 

Conversely, suppose that all the properties $(${\rm i}$) - (${\rm iv}$)$ hold. Let $P_w=\bigcup^{t^*}_{t=1}\{i_t,j_t\}$ where $\{i_t,j_t\}_{t=1,2,\ldots,t^*}$ is a maximal outer elimination sequence. We note that $P_w$ is empty if there is no elimination pair. We define the partition $\pi$ as follows:
First, for each $t=1,2,\ldots,t^*$, we set $\pi(i_t)=\{i_t,j_t\}$. Second, for each player $i \in R_w$, we set $\pi(i)=\{i\}$. Finally, we partition the players in $B_w$ and $S_w$ as follows. 
\begin{itemize}
\item If there is a player $i \in P_w$ who is a friend of every player in $B_w$, we put each $j \in B_w$ and the unique friend of $j$ in $S_w$ into a pair. 
\item Otherwise, all players in $B_w$ form a coalition and put each player in $S_w$ into a singleton. 
\end{itemize}

Since each player belongs to a clique or has at least two friends in his coalition, $\pi$ can be easily seen to be IR-robust under deletion of a single player. Now take any player $i \in N$. We will show that $i$ has no CIS-deviation even after removal of a single player. Now consider the following cases. 
\begin{itemize}
\item $i \in R_w$: By Lemma \ref{lem:degree} and $(${\rm iii}$)$, $i$ is an isolated vertex of the friendship graph and hence has no incentive to deviate to the other coalitions, even after deletion of a single player. 
\item $i \in B_w$: By definition, $i$ has at least two friends in his coalition, or $i$ forms a pair with his friend in $S_w$. If $\pi(i)$ contains at least two friends of $i$, $i$ has no CIS-deviation even after removal of an arbitrary single player, as there is at least one friend of $i$ in $\pi(i)$ who would be worse off by the deviation of $i$. 
Suppose $i$ forms a pair with his friend $j$ in $S_w$. In order for $i$ to have a CIS-deviation to other coalitions, $j$ must disappear; but then, every other coalition still contains an enemy of $i$, thereby implying that $i$ has no CIS-deviation to any other coalition even a single player disappears.  
\item $i \in S_w$: Note that $i$ is adjacent to only one player $j$ in $B_w$. Thus, $i$ has an incentive to deviate to another coalition $X \in \pi$ only if $X=B_w$. However, $j$ has at least two friends in $B_w$ and hence $B_w$ contains at least one enemy of $i$ even after removal of an arbitrary single player. Thus, $i$ has no CIS-deviation even after removal of a single player. 
\item $i \in P_w$: 
First, by Lemma \ref{lem} and $(${\rm iii}$)$, $i$ has no incentive to join a singleton included in $S_w$ or $R_w$. 
Second, player $i$ has no CIS-deviation to another elimination pair even after deletion of any single player, since by $(${\rm ii}$)$, at least one player in his coalition or the deviating pair would be worse off by the deviation of $i$. 
Finally, the property $(${\rm iv}$)$ further ensures that player $i$ has no CIS-deviation to a coalition containing a player $j \in B_w$. Indeed, player $i$ has a CIS-deviation to such coalition only when the other player of $i$'s coalition disappears; but then $j$'s coalition contains at least one enemy of $i$. Thus, $i$ has no CIS-deviation even after removal of a single player. 
\end{itemize}
The proof is complete. 
\end{proof}

\begin{figure}[htb]
\centering
\begin{tikzpicture}[scale=1, transform shape, every node/.style={minimum size=2mm, inner sep=1pt}]
	\node[draw, circle] at (5.1,1.3) {};
	\node[draw, circle] at (5.1,0.8) {};
	\node[draw, circle] at (5.1,0.3) {};
	
	\node[draw, circle](d) at (2,1.3) {};
	\node[draw, circle](aa) at (2.5,1.3) {};
	\node[draw, circle](bb) at (3,1.3) {};
	\node[draw, circle](cc) at (3.5,1.3) {};
	\node[draw, circle](e) at (4,1.3) {};

	\draw[-, >=latex,thick] (1.8,1.6)--(d) (aa)--(d) (bb)--(aa) (cc)--(bb) (cc)--(e) (e)--(4.2,1.6);
	
	\node[draw, circle](a) at (2.5,0.7) {};
	\node[draw, circle](b) at (3,0.7) {};
	\node[draw, circle](c) at (3.5,0.7) {};
	
	\node[draw,circle](1) at (1,1.2){};
	\node[draw,circle](2) at (0.5,1.2){};
	\node[draw,circle](3) at (1,0.85){};
	\node[draw,circle](4) at (0.5,0.85){};
	\node[draw,circle](5) at (1,0.5){};
	\node[draw,circle](6) at (0.5,0.5){};
	\draw[-,>=latex,thick] (1)--(2) (3)--(4) (5)--(6); 
	
	\node at (0.7,1.7) {\small $P_w$};
	\node at (3,1.7) {\small $B_w$};
	\node at (3,0.3) {\small $S_w$};
	\node at (5.3,1.7){\small $R_w$};
	
	\draw[-, >=latex,thick] (a)--(aa) (b)--(bb) (c)--(cc);
	
	\draw[-, >=latex,thick] (0,0)--(0,2) (0,0)--(6,0) (0,2)--(6,2) (6,2)--(6,0);
	\draw[-,>=latex,thick] (1.5,2)--(1.5,0) (4.5,2)--(4.5,0) (4.5,1)--(1.5,1); 
\end{tikzpicture}
\caption{A partition of the player set into $P_w,S_w,B_w,R_w$.
\label{fig:decomposition}
}
\end{figure}
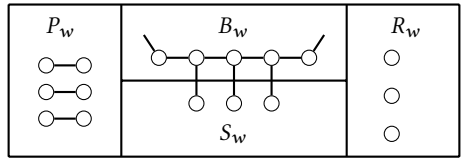

Building on the above characterization, it is easy to see that we can decide in polynomial time whether a symmetric friend-oriented game admits a CIS-robut outcome under deletion of a single player. The proof employs a simple procedure, which iteratively expands an outer elimination sequence and eventually decompose the player set into $P_w$, $S_w$, $B_w$, and $R_w$.

\begin{theorem}\label{thm:CIS:poly}
For a symmetric friend-oriented game and $k=1$, deciding the existence of a CIS-robust outcome can be done in polynomial time. 
\end{theorem}
\begin{proof}
Consider the following algorithm which takes as input a friendship graph $G_w$ and returns a CIS-robust outcome if it exists. The algorithm recursively finds an elimination pair and divides the players into the three sets $P$, $S$, $B$, and $R$. By Theorem \ref{thm:CIS}, the algorithm returns a CIS-robust partition if it exists. 
Indeed, if there is a CIS-robust outcome, the algorithm correctly computes $P=P_w$; and we have shown in the proof of Theorem \ref{thm:CIS}, the output $\pi$ of Algorithm \ref{alg:CIS} is CIS-robust under deletion of a single player. Conversely, if there is no CIS-robust outcome, $\pi$ is not CIS-robust and thus Algorithm \ref{alg:CIS} fails. Clearly, Algorithm \ref{alg:CIS} runs in polynomial time.

\begin{algorithm}                      
\caption{Algorithm for deciding the existence of CIS-robust outcomes}         
\label{alg:CIS}                          
\begin{algorithmic}[1]                  
\REQUIRE a friendship graph $G_w$ of a symmetric friend-oriented game $(N,w)$
\ENSURE CIS-robust partition $\pi$ if it exists
\STATE initialize $\pi \leftarrow \emptyset$
\STATE $P \leftarrow \emptyset$
\WHILE{there is a pair of players $i,j \in N \setminus P$ such that $(${\rm E1}$)$ $j$ is the unique friend of $i$ in $G_w \setminus P$ and $j$ is a pseudo-center in $G_w\setminus P$, or $(${\rm E2}$)$ $j$ is the unique friend of $i$ in $G_w \setminus P$, and $i$ is a friend of some player in $P$}
\STATE $\pi \leftarrow \pi \cup \{\{i,j\}\}$ and $P  \leftarrow P \cup \{i,j\}$
\ENDWHILE
\STATE set $S$ to be the set of players $j$ in $N \setminus P$ who have only one friend $f_j$ in $N \setminus P$. 
\STATE set $B$ to be the set of players who have at least two friends in $N \setminus (P \cup S)$. 
\STATE set $R \leftarrow N\setminus (P \cup S \cup B)$
\STATE set $\pi \leftarrow \pi \cup \{\, \{i\} \mid i \in R\,\}$
\IF{there is a player $i \in P$ who is adjacent to every player in $B$}
\STATE $\pi \leftarrow \pi \cup \{\, \{j,f_j\} \mid j \in S\,\}$
\ELSE
\STATE $\pi \leftarrow \pi \cup \{B\} \cup \{\, \{i\} \mid i \in S \,\}$
\ENDIF
\IF{$\pi$ is a CIS-robust partition}
\RETURN $\pi$
\ELSE
\RETURN {\em fail}
\ENDIF
\end{algorithmic}
\end{algorithm}
\end{proof}

The above result turns out to be tight in several aspects. We first show that for $k=2$, finding a CIS-robust outcome of a symmetric friend-oriented game is NP-hard. 

\begin{theorem}\label{thm:NPh:CIS}
For a symmetric friend-oriented game $(N,w)$, it is NP-complete to decide the existence of a CIS-robust outcome even for $k=2$. 
\end{theorem}
\begin{proof}
CIS-robustness can be verified in polynomial time: for each set $X \subseteq N$ of size at most two, one can check in polynomial time whether $\pi_{-X}$ is contractually individually stable. So our problem is in NP. To show hardness, we give a reduction from {\sc Exact-3-Cover (X3C)}. Recall that an instance of {\sc X3C} is given by a set of elements $V=\{v_1,v_2,\dots, v_{3r}\}$ and a family $\calS$ of three-element subsets of $V$; it is a `yes'-instance if and only if there is an {\em exact cover} $\calS' \subseteq \calS$ with $|\calS'|=r$ and $\bigcup_{S \in \calS'}S = V$. 

{\em Construction}: Given an instance $(V,\calS)$ of X3C, we construct an instance of a friend-oriented game as follows. For each $v \in V$, we create a {\em vertex player} $v$. For each vertex $v \in V$, we create a {\em vertex gadget} $G_v$, which enforces the corresponding vertex player $v$ to have at least two friends in his robust coalition. Specifically, $G_v$ consists of vertex player $v$, two friends $f^1_v$ and $f^2_v$ of $v$, and one enemy $e_v$ of $v$. All the three players $f^1_v$, $f^2_v$, and $e_v$ are friends to each other, $f^1_v$ and $f^2_v$ are enemies of all the vertex players except for $v$, and the player $e_v$ is an enemy of all the vertex players. Figure 4(a) illustrates $G_v$. For each $S=\{u,v,w\} \in \calS$, we create a {\em set gadget} $G_S$ which consists of its vertex players $u,v,w$, and cliques $\{S^1_v,S^2_v,S^3_v\}$ for $v \in S$. Specifically, $S^1_v$ and $S^2_v$ are a friend of $v$ for each $v \in S$; $S^3_u, S^3_v,S^3_w$ form a clique; and the pairs of $S^1_u$ and $S^1_v$, $S^2_u$ and $S^2_w$, and $S^1_w$ and $S^2_v$ are friends to each other. 
See Figure 4(b) for an illustration. Unless specified otherwise, players are enemies to each other. Finally, we set $k=2$. 

\begin{figure}[htb]
\begin{subfigure}[t]{0.4\columnwidth}
\centering
\begin{tikzpicture}[scale=0.6, transform shape, every node/.style={minimum size=8mm, inner sep=1pt}]
	\node[draw, circle, fill=gray!70](0) at (0,1) {$v$};
	\node[draw, circle](1) at (-1,0) {$f^1_{v}$};
	\node[draw, circle](2) at (1,0) {$f^2_{v}$};
	\node[draw, circle](3) at (0,-1) {$e_{v}$};
	\draw[-, >=latex,thick] (0)--(1);
	\draw[-, >=latex,thick] (0)--(2);
	\draw[-, >=latex,thick] (2)--(1);
	\draw[-, >=latex,thick] (2)--(3);
	\draw[-, >=latex,thick] (3)--(1);
\end{tikzpicture}
	\caption{Vertex gadget $G_v$}
	\label{gadget1}
\end{subfigure}%
\centering
\begin{subfigure}[t]{0.4\columnwidth}
\begin{tikzpicture}[scale=0.6, transform shape, every node/.style={minimum size=6mm, inner sep=1pt}]
	\node[draw, circle, fill=gray!70](u) at (-3.2,-0.5) {$u$};
	\node[draw, circle](u1) at (-2.3,0.5) {$S^1_u$};
	\node[draw, circle](u2) at (-2.3,-1.5) {$S^2_u$};
	
	\draw[-, >=latex,thick] (u)--(u2) (u)--(u1) (u2)--(u1);
	
	\node[draw, circle, fill=gray!70](v) at (0,3) {$v$};
	\node[draw, circle](v1) at (-1,2) {$S^1_v$};
	\node[draw, circle](v2) at (1,2) {$S^2_v$};
	
	\draw[-, >=latex,thick] (v)--(v2) (v)--(v1) (v2)--(v1);
	
	\node[draw, circle, fill=gray!70](w) at (3.2,-0.5) {$w$};
	\node[draw, circle](w1) at (2.3,0.5) {$S^1_w$};
	\node[draw, circle](w2) at (2.3,-1.5) {$S^2_w$};
	
	\draw[-, >=latex,thick] (w)--(w2) (w)--(w1) (w2)--(w1);
	
	\node[draw, circle](v3) at (0,1) {$S^3_v$};
	\node[draw, circle](u3) at (-1,-0.5) {$S^3_u$};
	\node[draw, circle](w3) at (1,-0.5) {$S^3_w$};
	
	\draw[-, >=latex,thick] (u3)--(v3) (u3)--(w3) (w3)--(v3);
	
	\draw[-, >=latex,thick] (u3)--(u1) (u3)--(u2);
	\draw[-, >=latex,thick] (w3)--(w1) (w3)--(w2);
	\draw[-, >=latex,thick] (v3)--(v1) (v3)--(v2);
	
	\draw[-, >=latex,thick] (u1)--(v1) (u2)--(w2) (v2)--(w1);
\end{tikzpicture}
\caption{Set gadget $G_S$ for $S=\{u,v,w\}$}
	\label{gadget2}
\end{subfigure}
\caption{Gadgets constructed in the proofs of Theorem~\ref{thm:NPh:CIS}. The grey nodes correspond to vertex players.}
\label{fig:NPh}
\end{figure}
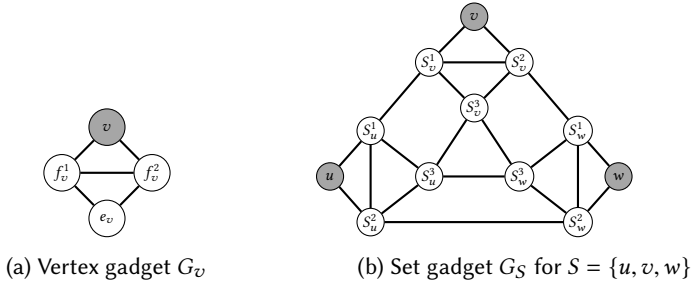



{\em Correctness}: 
Suppose that there is an exact cover $\calS' \subseteq \calS$. Then, we define $\pi$ as follows. For each set $S \in \calS'$ and each $v \in S$, we set $\pi(v)=\{v,S^1_v,S^2_v\}$; the remaining players of the set gadget forms a coalition, i.e., $\pi(S^3_u)=\{S^3_u,S^3_v,S^3_w\}$. For each $S \not \in \calS'$, the non-vertex players in the set gadget $G_s$ form a coalition, i.e., $\pi(S^1_u)=\{\, S^i_v \mid i=1,2,3 \land v \in S \,\}$. For each $v \in V$, we set $\pi(f^1_v)=\{f^1_v,f^2_v,e_v\}$. 
Since $\calS'$ is an exact cover, it can be easily verified that $\pi$ is a partition of the player set. Since each coalition is either a clique or a coalition with minimum degree at least three, $\pi$ is IR-robust under deletion of at most two players. To show that $\pi$ satisfies CIS-robustness, we note that in order for player $i$ to have a CIS-deviation to some coalition $X$, all her friends in $\pi(i)$ and all her enemies in $X$ must disappear. Thus, no player has a CIS-deviation to the coalitions $X$ of form $\{\, S^i_v \mid i=1,2,3 \land v \in S \,\}$ as every player not in $X$ has at least three enemies in $X$.  Also, no player has a CIS-deviation to the coalitions $X$ of form $\{v,S^1_v,S^2_v\}$ or $\{f^1_v,f^2_v,e_v\}$ as each player has at least two friends in his coalition and one enemy in $X$.  

To show the opposite direction, let $\pi$ be a CIS-robust partition. Consider a player $f^1_v$ for $v \in V$; we will show that $\pi(f^1_v)= \{f^1_v,f^2_v,e_v\}$. First, assume towards a contradiction that $v \in \pi(f^1_v)$. Recall that by Lemma \ref{lem:IR}, each player must be in a clique or in a coalition with minimum degree $k+1=3$. As there is no coalition containing $f^1_v$ with minimum degree three, $\pi(f^1_v)$ forms a clique of size at most three and we have either $\pi(f^1_v)=\{v,f^1_v\}$ or $\pi(f^1_v)=\{v,f^1_v,f^2_v\}$. When $\pi(f^1_v)=\{v,f^1_v\}$, player $f^2_v$ forms a singleton or a pair with $e_v$; however, either case admits the CIS-deviation of $f^1_v$ to $\pi(f^2_v)$ when player $v$ disappears, a contradiction. When $\pi(f^1_v)=\{v,f^1_v,f^2_v\}$, $e_v$ forms a singleton and hence player $e_v$ would have a CIS-deviation to $\pi(f^1_v)$ after player $v$ disappears, a contradiction. Thus, $v \not \in \pi(f^1_v)$ and hence $\pi(f^1_v) \subseteq \{f^1_v,f^2_v,e_v\}$. If $\pi(f^1_v) \neq \{f^1_v,f^2_v,e_v\}$, then there is a player $j \in \{f^1_v,f^2_v,e_v\}$ who forms a singleton and has a CIS-deviation to the other coalition of $i \in \{f^1_v,f^2_v,e_v\}$ with $i  \neq j$. Thus, for each $v \in V$, $\pi(f^1_v)= \{f^1_v,f^2_v,e_v\}$. 

This means that each vertex player $v$ needs to have at least two friends in his coalition as otherwise removing his enemy $e_v$ and one friend of $v$'s coalition would cause the CIS-deviation of $v$ to the remaining players $\{f^1_v,f^2_v\}$. Now the only IR-robust way to do this is to select triples of form $\{v,S^1_v,S^2_v\}$ where $S \in \calS$ and $v \in S$, and put them into a coalition. 

We will next show that if $\pi(v)=\{v,S^1_v,S^2_v\}$ for some $S=\{u,v,w\} \in \calS$, then it must be the case that $\pi(u)=\{u,S^1_u,S^2_u\}$ and $\pi(w)=\{w,S^1_w,S^2_w\}$. Suppose towards a contradiction that $\pi(v)=\{v,S^1_v,S^2_v\}$ for some $S=\{u,v,w\} \in \calS$ but $\pi(u)=\{u,T^1_u,T^2_u\}$ for some $T \in \calS$ with $T \neq S$. We note that $S^3_v$ needs to have at least two friends in his coalition; otherwise he would deviate to the coalition $\pi(v)$ after $v$ and one friend of $S^3_v$ in his coalition disappear. Thus, by Lemma \ref{lem:IR}, $S^3_v$ is in a clique of size three and it must be the case that $\pi(S^3_v)=\{S^3_u,S^3_v,S^3_w\}$. Then by Lemma \ref{lem:IR}, we have the following cases:
\begin{itemize}
\item $\pi(S^1_u)=\{S^1_u\}$ and $\pi(S^2_u) = \{S^2_u\}$: This would cause the CIS-deviation of $S^1_u$ to $\pi(S^2_u)$. 
\item $\pi(S^1_u)=\{S^1_u\}$ and $\pi(S^2_u) = \{S^2_u,S^2_w\}$: This would cause the CIS-deviation of $S^1_u$ to $\pi(S^2_u)$ after deletion of $S^2_w$. 
\item $\pi(S^1_u)=\{S^1_u,S^2_u\}$: This would cause the CIS-deviation of $u$ to $\pi(S^1_u)$ after removing two friends $T^1_u,T^2_u$ of $u$. 
\end{itemize}
In either case, we obtain a contradiction, and thus $\pi(u)=\{u,S^1_u,S^2_u\}$. Similarly, we have $\pi(w)=\{w,S^1_w,S^2_w\}$. Now let
\[
\calS'=\bigcup_{v \in V}\{\, S \in \calS \mid \pi(v)=\{v,S^1_v,S^2_v\}\,\}. 
\]
Clearly, each player $v \in V$ appears in some set in $\calS'$ Also, there is no pair of distinct sets $S,T \in \calS'$ that includes the same vertex player $v$, as otherwise, this would mean that $\pi(v)=\{v,S^1_v,S^2_v\}=\{v,T^1_v,T^2_v\}$, a contradiction. We conclude that $\calS'$ is an exact cover.
\end{proof}

A similar proof of Theorem \ref{thm:NPh:CIS} shows that finding an IS-robust outcome is NP-hard even if only a single player is allowed to disappear. 
\begin{theorem}\label{thm:NPh:IS}
For a symmetric friend-oriented game $(N,w)$, deciding the existence of an IS-robust outcome is NP-complete even for $k=1$. 
\end{theorem}
\begin{proof}
Clearly, our problem is in NP. Again, we give a reduction from {\sc Exact-3-Cover (X3C)}. 

{\em Construction}: Given an instance $(V,\calS)$ of X3C, we construct an instance of a friend-oriented game as follows. For each $v \in V$, we create a {\em vertex player} $v$. For each vertex $v \in V$, we create the same vertex gadget $G_v$ as in the proof of Theorem \ref{thm:NPh:CIS} with the additional {\em dummy} players $d^f_{v}$ and $d^e_v$, each of which is adjacent to $f^1_v$ and $e_v$, respectively. See Figure 5(a) for an illustration. For each $S=\{u,v,w\} \in \calS$, we create a {\em set gadget} $G_S$ which consists of its vertex players $u,v,w$, and cliques $\{S^1_v,S^2_v,S^3_v\}$ for $v \in S$. Specifically, $S^1_v$ and $S^2_v$ are a friend of $v$ for each $v \in S$, and $S^3_u, S^3_v,S^3_w$ form a clique. See Figure 5(b) for an illustration. Unless specified otherwise, players are enemies to each other. Finally, we set $k=1$. 

\begin{figure}[htb]
\begin{subfigure}[t]{0.4\columnwidth}
\centering
\begin{tikzpicture}[scale=0.6, transform shape, every node/.style={minimum size=8mm, inner sep=1pt}]
	\node[draw, circle, fill=gray!70](0) at (0,1) {$v$};
	\node[draw, circle](1) at (-1,0) {$f^1_{v}$};
	\node[draw, circle](2) at (1,0) {$f^2_{v}$};
	\node[draw, circle](3) at (0,-1) {$e_{v}$};
	
	\node[draw, circle](d1) at (-2.3,0) {$d^f_{v}$};
	\node[draw, circle](d3) at (0,-2.3) {$d^e_{v}$};
	\draw[-, >=latex,thick] (0)--(1);
	\draw[-, >=latex,thick] (0)--(2);
	\draw[-, >=latex,thick] (2)--(1);
	\draw[-, >=latex,thick] (2)--(3);
	\draw[-, >=latex,thick] (3)--(1);
	
	\draw[-, >=latex,thick] (1)--(d1) (3)--(d3); 
\end{tikzpicture}
	\caption{Vertex gadget $G_v$}
	\label{gadget1:IS}
\end{subfigure}%
\centering
\begin{subfigure}[t]{0.4\columnwidth}
\begin{tikzpicture}[scale=0.6, transform shape, every node/.style={minimum size=6mm, inner sep=1pt}]
	\node[draw, circle, fill=gray!70](u) at (-3,-0.5) {$u$};
	\node[draw, circle](u1) at (-2,0.5) {$S^1_u$};
	\node[draw, circle](u2) at (-2,-1.5) {$S^2_u$};
	
	\draw[-, >=latex,thick] (u)--(u2) (u)--(u1) (u2)--(u1);
	
	\node[draw, circle, fill=gray!70](v) at (0,3) {$v$};
	\node[draw, circle](v1) at (-1,2) {$S^1_v$};
	\node[draw, circle](v2) at (1,2) {$S^2_v$};
	
	\draw[-, >=latex,thick] (v)--(v2) (v)--(v1) (v2)--(v1);
	
	\node[draw, circle, fill=gray!70](w) at (3,-0.5) {$w$};
	\node[draw, circle](w1) at (2,0.5) {$S^1_w$};
	\node[draw, circle](w2) at (2,-1.5) {$S^2_w$};
	
	\draw[-, >=latex,thick] (w)--(w2) (w)--(w1) (w2)--(w1);
	
	\node[draw, circle](v3) at (0,1) {$S^3_v$};
	\node[draw, circle](u3) at (-1,-0.5) {$S^3_u$};
	\node[draw, circle](w3) at (1,-0.5) {$S^3_w$};
	
	\draw[-, >=latex,thick] (u3)--(v3) (u3)--(w3) (w3)--(v3);
	
	\draw[-, >=latex,thick] (u3)--(u1) (u3)--(u2);
	\draw[-, >=latex,thick] (w3)--(w1) (w3)--(w2);
	\draw[-, >=latex,thick] (v3)--(v1) (v3)--(v2);
	
\end{tikzpicture}
	\caption{Set gadget $G_S$ for $S=\{u,v,w\}$}
	\label{gadget2:IS}
\end{subfigure}
\caption{Gadgets constructed in the proofs of Theorem~\ref{thm:NPh:IS}. The grey nodes correspond to vertex players.}
\label{fig:NPh:IS}
\end{figure}
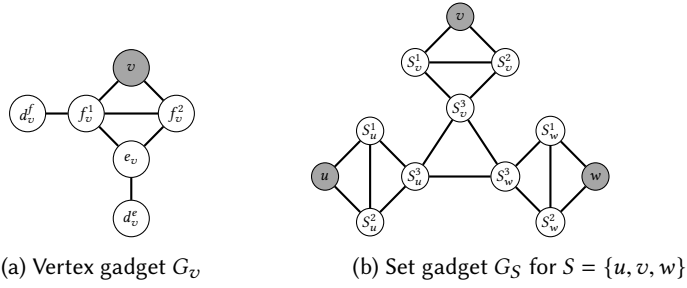

{\em Correctness}: 
Suppose that there is an exact cover $\calS' \subseteq \calS$. Again, we define $\pi$ as follows: 
\begin{itemize}
\item for each $S=\{u,v,w\} \in \calS'$, we set $\pi(u)=\{u,S^1_u,S^2_u\}$, $\pi(v)=\{v,S^1_v,S^2_v\}$, $\pi(w)=\{w,S^1_w,S^2_w\}$, and $\pi(S^3_u)=\{S^3_u,S^3_v,S^3_w\}$; and 
\item for each $S=\{u,v,w\} \not \in \calS'$, we set $\pi(S^1_u)=\{\, S^i_v \mid i=1,2,3 \land v \in S \,\}$; and
\item for each $v \in V$, we set $\pi(f^1_v)=\{f^1_v,f^2_v,e_v\}$, $\pi(d^f_{v})=\{d^f_{v}\}$ and $\pi(d^e_v)=\{d^e_v\}$.
\end{itemize}
Since $\calS'$ is an exact cover, it can be easily verified that $\pi$ is a partition of the player set. We will show that $\pi$ satisfies IS-robustness. No player has an IS-deviation to the coalitions $\{\, S^i_v \mid i=1,2,3 \land v \in S \,\}$ for $S \in \calS$, since each player outside the coalition has at least three enemies, and there is at least one player who would not accept such a deviation even after removal of a single player. Now fix $v \in V$. The only players who would have an incentive to deviate to the coalition $\{f^1_v,f^2_v,e_v\}$ are the vertex player $v$ or the dummy players $d^f_{v}$ and $d^e_v$. However, since $v$ has two friends in his coalition $\pi(v)$, and since there are two friends and one enemy of $v$ in $\{f^1_v,f^2_v,e_v\}$, $v$ would not have an IS-deviation even after removal of an arbitrary player. Also, no dummy player has an IS-deviation to $\{f^1_v,f^2_v,e_v\}$ even after removal of a single player, as at least one enemy of that player remains in the coalition. 
Further, the only player $i$ who has an incentive to deviate to a coalition $\pi(v)=\{v,S^1_v,S^2_v\}$ after removal of at most one player is $S^3_v$ or $f^j_v$ for some $j=1,2$. However, for $i=S^3_v$, there are two friends and one enemy of $i$ in $\pi(v)$, and two friends of $i$ in $\pi(i)$; thus, even after removal of $i$'s enemy from $\pi(v)$, or removal of $i$'s friends from $\pi(i)$, $i$ would not have an IS-deviation to coalition $\pi(v)$. Similarly, $i$ has no IS-deviation to $\pi(v)$ even after removal of a single player when $i=f^j_v$. Finally, it can be easily verified that no player has an IS-deviation to coalitions of dummy players even after removing an arbitrary single player. We conclude that $\pi$ is IS-robust. 

To show the opposite direction, let $\pi$ be an IS-robust partition. Again, consider player $f^1_v$ for $v \in V$; we will show that $\pi(f^1_v)= \{f^1_v,f^2_v,e_v\}$. First, assume towards a contradiction that $v \in \pi(f^1_v)$. Recall that by Lemma \ref{lem:IR}, each player must be in a pair of an edge or in a coalition with minimum degree two. Then, we have the following cases: 
\begin{itemize}
\item $\pi(f^1_v)=\{v,f^1_v\}$: Then, player $f^2_v$ forms a singleton or a pair with $e_v$, either of which admits an IS-deviation of $f^1_v$ to $\pi(f^2_v)$ when player $v$ disappears, a contradiction. 
\item $\pi(f^1_v)=\{v,f^1_v,f^2_v\}$: Then, $e_v$ forms a singleton or a pair with its neighbor, which means that player $e_v$ would have an IS-deviation to $\pi(f^1_v)$ after player $v$ disappears, a contradiction. 
\item $\{v,f^1_v,f^2_v\} \subseteq \pi(f^1_v)$ and $\pi(f^1_v)$ contains some player outside $G_v$: Then, $d^f_{v}$ forms a singleton and hence player $f^1_v$ would have an IS-deviation to $\pi(e_v)$ after player $v$ disappears, a contradiction. 
\item $\{v,f^1_v,f^2_v,e_v\} \subseteq \pi(f^1_v)$: Then, a neighbor $d^e_v$ of $e_v$ forms a singleton and hence player $e_v$ would have an IS-deviation to $\pi(d^2_v)$ after player $f^1_v$ disappears, a contradiction. 
\end{itemize}
In either case, we obtain a contradiction and thus $v \not \in \pi(f^1_v)$. Also, if $v \in \pi(f^2_v)$, then by Lemma \ref{lem:IR}, $\pi(f^2_v)=\{v,f^2_v\}$ and $\pi(f^1_v)$ is a coalition of size at most two, meaning that $f^1_v$ has an IS-deviation to $\pi(f^2_v)$ after the other player in her coalition disappears. Hence, $v \not \in \pi(f^2_v)$. 
If $\pi(f^1_v) \neq \{f^1_v,f^2_v,e_v\}$, then player $f^2_v$ either stays alone or forms a pair with her neighbor. In the former case, $f^2_v$ would have an IS-deviation to some neighboring coalition after one of the dummy players disappear. In the latter case, there is a player $j \in \{f^1_v,f^2_v,e_v\} \setminus \pi(f^2_v)$ who forms a singleton or a pair, and has an IS-deviation to coalition $\pi(f^2_v)$. Thus, for each $v \in V$, $\pi(f^1_v)= \{f^1_v,f^2_v,e_v\}$, and hence $v$ needs to have at least two friends in his coalition as otherwise removing his enemy $e_v$ would cause the IS-deviation of $v$ to the remaining players $\{f^1_v,f^2_v\}$. Now the only IR-robust way to do this is to select triples of form $\{v,S^1_v,S^2_v\}$ where $S \in \calS$ and $v \in S$, and put them into a coalition. 

We will now show that if $\pi(v)=\{v,S^1_v,S^2_v\}$ for some $S=\{u,v,w\} \in \calS$, then it must be the case that $\pi(u)=\{u,S^1_u,S^2_u\}$ and $\pi(w)=\{w,S^1_w,S^2_w\}$.
Now suppose towards a contradiction that $\pi(v)=\{v,S^1_v,S^2_v\}$ for some $S=\{u,v,w\} \in \calS$ but $\pi(u)=\{u,T^1_u,T^2_u\}$ for some $T \in \calS$ with $T \neq S$. We note that $S^3_v$ needs to have at least two friends in his coalition; otherwise he would have an IS-deviation to the coalition $\pi(v)$ after $v$ disappears. Thus, by Lemma \ref{lem:IR}, it must be the case that $\pi(S^3_v)=\{S^3_u,S^3_v,S^3_w\}$. Then by Lemma \ref{lem:IR}, we have the following cases:
\begin{itemize}
\item $\pi(S^1_u)=\{S^1_u\}$ and $\pi(S^2_u)=\{S^2_u\}$: This would cause the IS-deviation of $S^1_u$ to $\pi(S^2_u)$. 
\item $\pi(S^1_u)=\{S^1_u,S^2_u\}$: This would cause the IS-deviation of $u$ to $\pi(S^1_u)$ after removing a friend $T^1_u$ of $u$. 
\end{itemize}
In either case, we obtain a contradiction, and thus $\pi(u)=\{u,S^1_u,S^2_u\}$. Similarly, we have $\pi(w)=\{w,S^1_w,S^2_w\}$. Now let
\[
\calS'=\bigcup_{v \in V}\{\, S \in \calS \mid \pi(v)=\{v,S^1_v,S^2_v\}\,\}. 
\]
Clearly, $\calS'$ is a cover. Also, there is no pair of distinct sets $S,T \in \calS'$ that includes the same vertex player $v$, as otherwise, this would mean that $\pi(v)=\{v,S^1_v,S^2_v\}=\{v,T^1_v,T^2_v\}$, a contradiction. We conclude that $\calS'$ is an exact cover.
\end{proof}

We have not been able to identify whether the problem of computing a CR-robust outcome is polynomial-time solvable; we leave this question for future work.

\section{IR-robustness}
In Example \ref{ex:friend:nonexistence}, we have seen that an outcome of a hedonic game can fail to preserve some stability properties, even when players' preferences are symmetric friend-oriented. Our next question is the following: is it still possible to guarantee a minimum stability requirement, i.e., individual rationality, under deletion of players, while ensuring desirable property of the original partition? The answer is positive for individual stability when players have symmetric additively separable preferences. In these games, one can guarantee the existence of an individually stable partition that is IR-robust. In outline, the algorithm works as follows: starting with all-the singleton partition, we keep letting a player who has an IS-deviation deviate to his preferred coalition until there does not remain such a player. The algorithm correctly identifies an outcome that is individually stable and IR-robust due to the potential argument and the fact that each player only accepts a player for whom she has a non-negative utility.

\begin{theorem}\label{thm:IS-IRrobust}
For any symmetric additively separable game $(N,w)$ and any natural number $k>0$, there exists an individually stable and IR-robust partition. 
\end{theorem}
\begin{proof}
Consider the following algorithm $\calA$: starting with all-the singleton partition $\pi$, if there is a player $i$ who has an IS-deviation to some coalition $S \in \pi$, we keep letting $i$ deviate to his preferred coalition $S$, i.e., set $\pi = (\pi \setminus \{\pi(i),S\}) \cup \{\pi(i)\setminus \{i\},S\cup \{i\} \}$, until there does not remain such a player. \citet{Bogomolnaia2002} proved that the algorithm terminates in finite steps due to the potential argument. Hence the resulting partition $\pi$ is individually stable because no player has an IS-deviation at the termination. We will show that no coalition includes a pair of players who are enemies to each other at any point of the execution of the algorithm, which implies that $\pi$ is IR-robust. First, it is clear that the claim holds for the initial partition. Suppose that until step $t$ of the execution, no coalition in $\pi^t$ includes a pair of players who are enemies to each other; we will prove that it still holds for $t+1$. Let $i$ be a player who deviates to coalition $S$ from step $t$ to step $t+1$. If there is a player $j$ in $S$ with negative utility $w(j,i)<0$, then she would not accept a deviation of $i$ to $S$, implying that all players in $S$ have non-negative utility for $i$ and vice versa. Hence players in each of the coalitions like each other. 
\end{proof}

In general, finding an individually stable outcome of a symmetric additively separable game is known to be computationally intractable \citep{Gairing2011}. 
In contrast, we can efficiently construct an individually stable partition that is IR-robust in symmetric friend-oriented games, in which each weight only takes two values. 

\begin{theorem}\label{thm:sF:IS-IRrobust}
For any symmetric friend-oriented $(N,w)$ and any natural number $k>0$, one can compute an individually stable and IR-robust partition in polynomial time. 
\end{theorem}
\begin{proof}
We denote by $w(T)$ the total number of friendship pairs in each coalition $T$, i.e., $w(T)= |\{\,\{i,j\} \mid i,j \in T \land w(i,j)>0\,\}|$.
To show that Algorithm $\calA$ terminates in polynomial time, we define the potential function $\Phi(\pi)$ to be the total number of friendship pairs within the coalitions, i.e., $\Phi(\pi):= \sum_{T \in \pi}w(T)$, for partition $\pi$. First, observe that the value of the potential function is bounded up to the total number of friendship pairs, namely, $\Phi(\pi) \le m$. Moreover, at each step of the algorithm $\Phi(\pi)$ strictly increases at least by $1$. To see this, let $i$ be a player who deviates from his own coalition to another coalition $S$ from step $t$ to step $t+1$; denote by $\pi'$ the resulting partition, i.e., $\pi'=(\pi \setminus \{\pi(i),S\}) \cup \{\pi(i)\setminus \{i\},S\cup \{i\}\}$. Recall that in the previous proof for Theorem \ref{thm:IS-IRrobust}, we have seen that each coalition of $\pi$ is a clique. Now, since $i$ has more friends in $S$ than in $\pi(i)$ and all players in $S$ accept $i$'s deviation, we have
\begin{align*}
\Phi(\pi') - \Phi(\pi) &= \sum_{T \in \pi'}w(T)-\sum_{T \in \pi}w(T), \\
&= |S| - |\pi(i) \setminus \{i\}|>0.
\end{align*}
This implies that Algorithm $\calA$ terminates in polynomial time.
\end{proof}

We note that without symmetry, the set of outcomes that are both individually stable and IR-robust can be empty. 
\begin{example}\label{ex:friend}
\upshape
Consider a friend-oriented game with four players $a,b,c,d$ where all players $a,b,c$ consider the others to be a friend, but $d$'s friend is player $c$ only. Figure \ref{fig:friend:nonsymmetric} illustrates the corresponding friendship graph. 
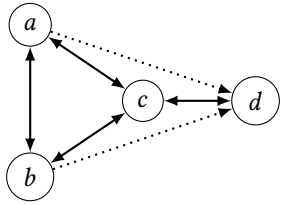
\begin{figure}[htb]
\centering
\begin{tikzpicture}[scale=1, transform shape]
	\node[draw, circle](1) at (-1,1) {$a$};
	\node[draw, circle](2) at (-1,-1) {$b$};
	\node[draw, circle](3) at (0.5,0) {$c$};
	\node[draw, circle](4) at (2,0) {$d$};	
	\draw[<->, >=latex,thick] (1)--(2);
	\draw[<->, >=latex,thick] (2)--(3);
	\draw[<->, >=latex,thick] (1)--(3);
	\draw[<->, >=latex,thick] (4)--(3);
	\draw[->, >=latex,thick,dotted] (1)--(4);
	\draw[->, >=latex,thick,dotted] (2)--(4);
\end{tikzpicture}
\caption{Non-existence of an IS and IR-robust partition for a non-symmetric friend-oriented game. The non-symmetric relations are represented by the dotted lines.
\label{fig:friend:nonsymmetric}
}
\end{figure}

Suppose towards a contradiction that there is an individually stable partition $\pi$ that is IR-robust. Then, by IR-robustness, $d$ must be in a pair with his unique friend $c$, or stay alone. In the former case, $\pi$ would not be individually stable, since if $a$ and $b$ stay alone, then $a$ would have an IS-deviation to the coalition $\{b\}$, and since if $a$ and $b$ stay together, then $c$ would have an IS-deviation to the pair $\{a,b\}$, a contradiction. Similarly, in the latter case, players $a,b,c$ must form the same coalition together by individual stability, but then player $d$ would have an IS-deviation to that coalition, a contradiction. 
\end{example}

\section{Conclusion}
We believe that this paper has made a first important step towards a future stream of research, sparked by the chemistry of two concepts, robustness and stability. Below, we list several interesting questions for future work.

Most obviously, while our main focus was on robustness against agents' non-participation, studying other types of robustness would be an important topic of research. For instance, one might want to consider sudden failure of agents' friendship relations, due to individual or political conflicts. There are also further classes subclasses of additively separable games that we have not considered in this paper, most notably fractional hedonic games \citep{aziz2014fractional,AzizBBHOP17}, which one can study from both existence and complexity aspects; in particular, it would be interesting to investigate whether a similar graph-theoretic characterization of friendship graphs that ensure the existence of stable outcomes can be obtained.
 
Further, the definition of robustness in this work only considers agent failure in a uniform and deterministic sense. However, one might want to consider specific coalitional failure, rather than all of them; for example, our model can be extended to the probabilistic setting where each agent may have different probability of not participating the game. Finally, it would be interesting to extend this line of work to other settings where stability plays an important role. Examples include stable marriage problem \citep{GaleShapley} and group activity selection problem \citep{Darmann2012,igarashi2017gasp}.

\bibliographystyle{ACM-Reference-Format}
\bibliography{abb,hedonic,resilient}
\end{document}